\documentclass{article}
\usepackage{graphicx}
\usepackage{amsmath}
\usepackage{amssymb}
\usepackage{amsthm}

\usepackage{authblk}
\usepackage[small]{caption}
\usepackage{xspace}
\usepackage{url}
\usepackage{algorithmic}
\usepackage{algorithm}
\usepackage[T1]{fontenc}
\usepackage[adobe-utopia]{mathdesign}
 \newtheorem{theorem}{Theorem}
\newtheorem{problem}{Problem}

\newcommand{\ignore}[1]{}
\newcommand{\ie}{\textit{i.e.}\xspace}
\newcommand{\eg}{\textit{e.g.}\xspace}
\newcommand{\ea}{\textit{et al.}\xspace}

\begin{document}

\title{An Exact Algorithm for Side-Chain Placement\newline in Protein Design}
\author[1]{Stefan Canzar}
\author[2]{Nora C.\ Toussaint}
\author[1]{Gunnar W.\ Klau}

\affil[1]{\small CWI, Life Sciences group, Science Park 123, 1098 XG
  Amsterdam, the Netherlands, \{stefan.canzar, gunnar.klau\}@cwi.nl}
\affil[2]{University of T{\"u}bingen, Center for Bioinformatics,
            Applied Bioinformatics, Sand 14, 72076 T{\"u}bingen,
            Germany, toussaint@informatik.uni-tuebingen.de}
\date{}


\maketitle

\begin{abstract}
Computational protein design aims at constructing novel or improved functions on the
structure of a given protein backbone and has important applications
in the pharmaceutical and biotechnical industry. The underlying
combinatorial side-chain placement problem consists of choosing a 
side-chain placement for each residue position such that the
resulting overall energy is minimum. The choice of the side-chain then
also determines the amino acid for this position. 
Many algorithms for this $\mathcal{NP}$-hard problem have been proposed in the context of homology
modeling, which, however, 
reach their limits when faced with large protein design instances.
 
In this paper, we propose a new exact method for the side-chain placement problem
that works well even for large instance sizes as they appear in
protein design. Our main contribution is a dedicated
branch-and-bound algorithm that combines tight upper and lower
bounds resulting from a novel Lagrangian relaxation approach for
side-chain placement. Our experimental results show that our method
outperforms alternative state-of-the-art exact approaches and makes it
possible to optimally solve large protein design instances routinely.
\end{abstract}

\section{Introduction}
\label{sec:intro}

Protein design aims at constructing novel or improved functions on the
structure of a given protein backbone.  
Since proteins are key players in virtually all biological processes, the ability to design proteins is of 
great practical interest, \eg, to the pharmaceutical and biotechnological industry. Experimental protein design methods, 
such as directed evolution \cite{Bloom2005}, have been applied successfully. However, since experimental methods are time- and money-consuming, computational approaches are an attractive alternative.

Computational protein design is related to the side-chain placement (SCP) problem in protein homology modeling. Given the modeled backbone of a protein, 
the amino acid side-chains have to be placed on this backbone in the energetically most favorable conformation. Two assumptions are commonly made: 
(i) side-chains adopt only statistically dominant low-energy side-chain conformations, the so-called rotamers \cite{Dunbrack2002}, and 
(ii) the energy of a protein is the sum of intrinsic
side-chain energies and pairwise interaction energies. 
These assumptions lead to the following discrete optimization problem: For each residue position choose a rotamer such that the total 
energy of the protein is minimum. This problem has been shown to be NP-hard \cite{Pierce2002} and inapproximable \cite{Chazelle2004}. 


In protein design the candidate rotamers at each position do not only come from a single amino acid but from several potential amino acids, 
yielding very large problem instances. Previous \textit{in silico} approaches to protein design differ in their choice of rotamer library, 
energy function, and optimization method. Utilization of a higher-resolution rotamer library and a more accurate energy function will improve the results. 
On the other hand, it will increase computation time and problem size. Regarding the optimization methods, computational protein design approaches 
generally employ computationally expensive
heuristics such as the Monte Carlo method \cite{Dantas2003,Dantas2007,Shah2007}. 
Other heuristics, which have been proposed for SCP in protein homology modeling, could also be applied 
\cite{Wernisch2000,Xiang2001,Desmet2002,Yanover2006,Sontag2008}. 
However, Voigt \textit{et
  al.} \cite{Voigt2000} have shown that these inexact 
algorithms become less accurate with increasing problem size. Thus, exact methods capable of solving large protein 
design instances are desirable. Several approaches to solving the SCP problem exactly have been proposed, 
including dead end elimination \cite{Desmet1992,Goldstein1994,Pierce2000} (combined with systematic search \cite{Leach1998} 
or residue reduction \cite{Xie2006}), integer linear programming \cite{Althaus2002,Kingsford2005}, 
branch-and-bound \cite{Canutescu2003,Wernisch2000}
and tree decomposition \cite{Xu2006}. 
While most of these approaches work well for homology modeling, 
they reach their limits when applied to protein design.


In this paper, we propose a novel exact method for SCP
that works well even for large instance sizes as they appear in
protein design. After presenting the combinatorial problem formally in
Section~\ref{sec:comb-probl-form}, we describe our new method in
Section~\ref{sec:approach}. Our main contribution is a dedicated
branch-and-bound algorithm that combines tight upper and lower 
bounds resulting from a novel Lagrangian relaxation approach for
SCP\@. In Section~\ref{sec:results} we present and
discuss our experimental results, in which we show that our method
outperforms alternative state-of-the-art exact approaches
and makes it possible to optimally solve large protein design
instances routinely.

\section{Combinatorial Problem Formulation and Notation}
\label{sec:comb-probl-form}

We study the following graph-theoretic formulation of the side-chain placement problem:

\begin{problem}[SCP]
  Given a $k$-partite graph $G = (V, E)$, $V = V_1 \cup \ldots \cup
 V_k$, with node costs $c_v$, $v \in V$, and edge costs $c_{uv}$, $uv
  \in E$, determine an assignment $a: \{1,
  \ldots, k\} \to V$ with $a(i) \in V_i$, $1 \leq i \leq k$,
  such that the cost
\[
\sum_{i = 1}^k c_{a(i)} + \sum_{i = 1}^{k-1}\sum_{j=i+1}^k c_{a(i)a(j)}
\]
of the induced graph is minimum.  
\end{problem}

Here, each node set $V_i$ corresponds to the candidate rotamers for
the residue set at position $i$. Node costs model self energies of
rotamers and edge costs model interaction energies between pairs of
rotamers. A solution is given by selecting for each residue position $i$, $1
\leq i \leq k$, exactly one rotamer $a(i)$. Clearly, the choice of the
rotamer determines also the amino acid at this position.  

In the description of our algorithm we will also use a function $r: V
\to \{1, \ldots, k\}$ that denotes the residue position of a rotamer
$v$, that is, $r(v) = i$ if and only if $v \in V_i$.

\section{Lagrangian Relaxation Based Branch-and-Bound}
\label{sec:approach}

We now present our novel approach to solve the SCP problem to provable optimality. Its
core is the computation of sharp upper and lower bounds using
a novel Lagrangian relaxation technique within a dedicated branch-and-bound approach. 
 


\subsection{Upper and Lower Bounds by Lagrangian Relaxation}\label{sec:lagrrel}


Our relaxation
builds on an integer linear programming
(ILP) formulation for SCP that has been
introduced by Althaus \ea \cite{Althaus2002} and extended by
Kingsford \ea \cite{Kingsford2005}:

\begin{align}
\min \:& \sum_{v \in V}c_vx_v + \sum_{uv \in E'} c_{uv}y_{uv}\label{eq:start:kingsford}\\
\text{s.t.~}& \sum_{v \in V_i}x_v = 1 & 1 \leq i \leq k\label{eq:1}\\
& \sum_{u \in V_i}y_{uv} = x_v & 1 \leq i \leq k \text{, for all } v
\in V_j \text{ with } j \in \mathcal{N}^+_i\label{eq:2}\\
& \sum_{u \in V_i\atop c_{uv} < 0}y_{uv} \leq x_v & 1 \leq i \leq k \text{, for all } v
\in V_j \text{ with } j \notin \mathcal{N}^+_i\label{eq:3}\\
& x_v \in \{0, 1\} & \text{for all } v \in V\\
& y_{uv} \in \{0, 1\} & \text{ for all } uv \in E'\label{eq:end:kingsford}
\end{align}

The formulation contains binary variables $x_v$, for nodes $v \in V$, and
$y_{uv}$, for edges $uv \in E$, with the interpretation that a variable is $1$
if the corresponding node or edge is part of the induced subgraph and
$0$ otherwise. Constraints~\eqref{eq:1} express that exactly one
rotamer must be chosen per residue position. Constraints~\eqref{eq:2} and~\eqref{eq:3}
link node and edge variables. When a rotamer $v$ of a residue position $j$ is
chosen, \ie, $x_v = 1$, exactly one incident edge from each other
residue position $i \neq j$ must be chosen as well for residue positions $i$ that share
positively weighted edges with residue position $j$, \ie, $j \in
\mathcal{N}^+_i := 
\{\ell \in \{1, \ldots, k\} \mid \exists uv \in E, c_{uv} > 0, u \in
 V_i, v \in V_{\ell} \}$. If the two residue positions $i$ and $j$ are linked
  only by non-positive edges, \ie, $j \notin \mathcal{N}^+_i$, the
  relaxed constraints~\eqref{eq:3} apply: zero-weighted edges do not
  have to be forced to be in the solution and the corresponding
  variables can be removed from the ILP\@. Let $E' := E \setminus \{uv
  \in E \mid u \in V_i, v \in V_j, j \notin \mathcal{N}^+_i, c_{uv} = 0\}$
  be the set of remaining edges. 

  The distinction between pairs of residues $i,j$ with $j\in \mathcal{N}^+_i$ and pairs $i,j$ with
  $j\notin \mathcal{N}^+_i$ in constraints~\eqref{eq:2} and~\eqref{eq:3} leads to a considerably smaller number of 
  variables in practice and to a much better performance. Nevertheless, it is not crucial for the understanding	
  of our approach and we thus drop this distinction in the remainder of this work and treat all pairs of 
  residues as in constraint~\eqref{eq:2} for the sake of clarity of the
  presentation, that is, without removing any variables. In our
  implementation, however, we treat constraints~\eqref{eq:2}
  and~\eqref{eq:3} differently. 
         
  While previous work \cite{Althaus2002,Kingsford2005} focuses on
  solving the linear programming (LP) relaxation of
  \eqref{eq:start:kingsford}--\eqref{eq:end:kingsford}, we propose a
  Lagrangian relaxation approach. In the case of SCP
  this leads to a much
  more efficient algorithm, because we exploit structural knowledge
  of the SCP problem. The idea of Lagrangian relaxation is
  to relax constraints of an intractable problem, \eg, the SCP ILP, 
	such that the relaxed problem can be solved
  efficiently. The relaxed constraints are moved to the objective
  function, penalized by so-called Lagrangian multipliers. An optimal
  solution of the original problem, \ie, an energy-minimum choice of
  candidate rotamers, is also a solution of the relaxed problem, and
  every optimal solution of the relaxed problem provides a lower
  bound on the optimal score of the original problem. The Lagrangian multipliers are
 adjusted iteratively such that the lower bound increases
 gradually. Also, after each iteration, we can evaluate the solution
 of the relaxed problem and thus obtain a new upper bound to the SCP
 problem. During the iterative process, the lowest upper and
 highest lower bound move
 closer and closer together. If they coincide,
  a provably optimal SCP has been found. Otherwise,
  we stop the process after a fixed number of iterations and use the
  bounds within the branch-and-bound framework. 

The key idea of our Lagrangian relaxation approach is to define a total
order, denoted by $<$, on the residue positions, to split the constraints that link node and
edge variables into a left and a right part and then relax the
right part of the constraints. W.l.o.g.\ and for ease of notation
we assume that the residue positions have already been ordered, that is,
residue position $i$ denotes the $i$th residue position according to $<$.

First, we rewrite constraints~\eqref{eq:2} as 
left and right parts, that is,~\eqref{eq:2} becomes
\begin{align}
& \sum_{u \in V_i}y_{uv} = x_v & 1 \leq i \leq k-1 \text{, for all } v
\in V_j \text{ with } j > i\label{eq:4}\\ 
& \sum_{u \in V_i}y_{uv} = x_v & 2 \leq i \leq k \text{, for all } v
\in V_j \text{ with } j < i\label{eq:5}
\end{align} 
\ignore{
and~\eqref{eq:3} becomes
\begin{align}
& \sum_{u \in V_i\atop c_{uv} < 0}y_{uv} \leq x_v & 1 \leq i \leq k-1 \text{, for all } v
\in V_j \text{ with } j \notin \mathcal{N}^+_i \text{ and } j > i\label{eq:6}\\
& \sum_{u \in V_i\atop c_{uv} < 0}y_{uv} \leq x_v & 2 \leq i \leq k \text{, for all } v
\in V_j \text{ with } j \notin \mathcal{N}^+_i \text{ and } j < i\label{eq:7}
\end{align}
}
We dualize constraints~\eqref{eq:5} for all non-neighboring residues,
\ie, constraints for which $i>j+1$,  with Lagrangian multipliers $\lambda_v^i$.
To simplify notation, we further introduce $\lambda_v^i:=0$ for
neighboring residues, \ie, for which $r(v)+1=i$ holds.
For fixed Lagrangian multipliers $\lambda_v^i$ we obtain the following relaxation, which we denote by $(LR_\lambda)$:
\begin{align}
\min \:& \sum_{j = 1}^{k}\sum_{v \in V_j}\left(c_v + \sum_{i >
    j+1}\lambda^i_v\right )x_v + \sum_{uv \in
  E \atop r(u)<r(v)}\left(c_{uv} - \lambda_u^{r(v)}\right) y_{uv}\hspace*{-25em}\label{eq:start:relaxation}\\
\text{s.t.~}& \sum_{v \in V_i}x_v = 1 & 1 \leq i \leq k\label{eq:8}\\
& \sum_{u \in V_i}y_{uv} = x_v & \hspace*{-1em}1 \leq i \leq k-1 \text{, for all } v
\in V_j \text{ with } j > i\label{eq:9}\\ 
& \sum_{u \in V_i}y_{uv} = x_v & 2 \leq i \leq k \text{ for all } v
\in V_{i-1}\label{eq:10}\\
&x_v, y_{uv} \in \{0, 1\}&
\end{align}
Note that a distinction between pairs of residues according to constraints~\eqref{eq:2} and~\eqref{eq:3} requires the
Lagrangian multipliers $\lambda_v^i$ associated with constraints \eqref{eq:3} to be restricted in sign in order to guarantee that an optimal 
solution to $(LR_\lambda)$ yields a lower bound on the optimal score
of the \textsc{SCP} problem. 

An integral variable assignment that satisfies constraints~\eqref{eq:8}, \eqref{eq:10},
and constraints~\eqref{eq:9} for neighboring residues, \ie, for $j=i+1$, encodes a
path $p$ in the corresponding $k$-partite graph from a node in $V_1$ to a node in $V_k$ that traverses exclusively edges between neighboring residues. The 
remaining constraints of \eqref{eq:9} involve $y$-variables that do not appear in any other constraint and can thus be chosen independently of each other. 
In other words, we can determine the best possible contribution of a vertex $v$ to the overall objective value, under the assumption that $v$ lies 
on path $p$, by simply picking for every residue $i<r(v)-1$ the edge of minimum weight between $v$ and a node in $V_i$. More formally, we define the \emph{profit}
$\delta$ of a node $v$ as
\[
 \delta(v)=(c_v + \sum_{i>r(v)+1}\lambda_v^i) + \sum_{i=1}^{r(v)-2} \min_{u\in V_i}(c_{uv}-\lambda_u^{r(v)})\enspace,
\]
where the first term in brackets denotes the coefficient of variable $x_v$ in the objective function.
Then the score of a feasible solution to $(LR_\lambda)$ that induces a path $p=(v_1,v_2,\dots,v_k)$ with $v_i\in V_i$ in graph $G$
is \[
    \sum_{i=1}^k \delta(v_i) + \sum_{i=1}^{k-1} c_{v_iv_{i+1}}\enspace.
   \]
Let graph $G'$ be derived from $G$ by removing all edges between
non-neighboring residues, \ie, edges $u'v'$ with $|r(u')-r(v')|>1$, and 
by defining the weights of the remaining edges $uv$ as
$c_{uv}+\delta(v)$. Then, an optimal solution to $(LR_\lambda)$
corresponds to a shortest path in $G'$ from a node in $V_1$ to a node
in $V_k$, see Figure~\ref{fig:shortest_path}.

\begin{theorem}
 An optimal solution to $(LR_\lambda)$ can be computed in time $\mathcal{O}(|V|^2)$.
\end{theorem}

\begin{proof}
The profits of all nodes can clearly be computed in time $\mathcal{O}(|V|^2)$. Graph $G'$ is acyclic and thus a shortest path
can be computed in time linear in the number of edges in $G'$, \ie, $\mathcal{O}(\sum_{i=1}^{k-1}|V_i|\cdot|V_{i+1}|)$.
Note that a topological sorting of the vertices is implicitly given
by the $k$-partition. 
\end{proof}

\begin{figure}[tbp]
  \centering
  \includegraphics[width=7cm]{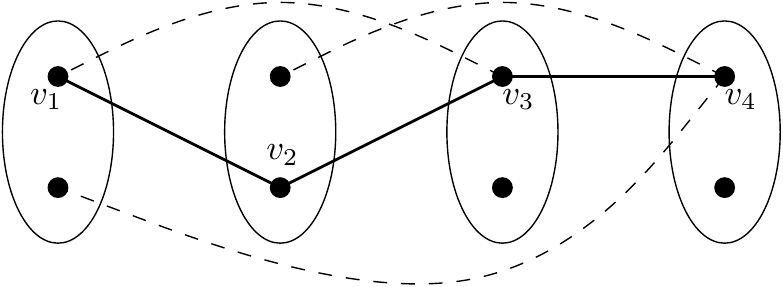}
  \caption{The structure of a feasible solution to $(LR_\lambda)$. The polygon drawn in solid line denotes the corresponding path $p=(v_1,v_2,v_3,v_4)$. Every 
node on path $p$ has exactly one incident edge to every residue left of it, except to its direct neighbor, depicted by the dashed lines.}
  \label{fig:shortest_path}
\end{figure}

We apply a standard \emph{subgradient optimization} technique \cite{HelKar70} to find those Lagrangian
multipliers $\lambda_v^i$ that yield the largest lower bound to our relaxation. This
iterative adaption of the Lagrangian multipliers only requires the profits
of a small fraction of the vertices to be recomputed from scratch in each iteration.
In practice, the shortest path computation for given profits by a simple dynamic programming
scheme dominates the overall running time needed to resolve the Lagrangian relaxation $(LR_\lambda)$ for modified multipliers $\lambda_v^i$.
In other words, the running time will be linear in the number of edges in $G'$ rather than $G$.


In order to improve the practical performance of our approach we sort
the residues by increasing number of rotamers. This ordering results
in a minimum total number of dualized constraints. 

\subsection{Branch-and-Bound}

We embed our Lagrangian bounding scheme into a branch-and-bound framework to obtain 
an energy-minimum rotamer assignment. The general idea is to divide the overall SCP problem into 
easier subproblems by fixing rotamers at individual residue positions and to solve the resulting problems
recursively. To avoid a complete enumeration of all possible rotamer assignments, we employ our
Lagrangian bounds to prune large parts of the enumeration tree. In particular, let $S_k$ be a subproblem
in which certain residue positions have been assigned a rotamer, and let $\hat{x}$ be an arbitrary solution to the
original SCP problem. If the lower bound $\underline{z}^k$ on the minimal energy assignment
for subproblem $S_k$ is larger than $z(\hat{x})$, \ie, the total energy of rotamer assignment $\hat{x}$, 
then no optimal solution to the original problem can be obtained from $S_k$ and we can prune the subtree
rooted at $S_k$ from the search space.  

\paragraph{Branching scheme.}

Starting from an SCP problem instance $S$, a straightforward branching rule
is to impose the constraint $x_{ij}=1$ in the left
child node of the search tree, and $x_{ij}=0$ in the right, \ie, fixing
and forbidding a rotamer $j$ at residue position $i$, 
respectively. However, since every rotamer $j$ is contained in a constraint~\eqref{eq:1} for
a residue position $i$, this leads to an unbalanced search tree, because the right child node leaves
$k-1$ possible rotamers for residue position $i$, whereas the left
child leaves only one possibility. While partitioning the
set of rotamers of a given residue position into two roughly equally
sized sets avoids this imbalance, we experienced,
however, a significantly smaller number of nodes in the tree with the
following, alternative
branching scheme. Instead of creating two subproblems, we create one
for each rotamer of a selected residue position $i$, by
fixing rotamer $j_k$ in subproblem $k$. Only for very large design instances, we first partitioned the sets of 
possible rotamers that were larger than some threshold $p$, into two smaller sets. 
The effectiveness of this scheme is mainly based on two
properties. First, when fixing a rotamer $j$ for a residue
position $i$,
we reduce the problem instance by incorporating the interaction
energy between rotamer $j$ and any other rotamer $j'$ of all residue positions
$i'\neq i$ into the self energy of $j'$. 
In contrast, in our
relaxation only the profits $\delta$ of rotamers of residue positions $i'>i$ would take into
account a further subdivision of the set of rotamers of residue
position $i$. The rotamer assignment to residue positions $i'<i$ would still rely on a 
correct choice of the incoming edges for the remaining rotamers of
residue position $i$, which could only be accomplished by iteratively adapting the
Lagrangian penalties. Second, a large enough set of child nodes
will give our depth-first search traversal of the branch-and-bound tree the freedom
to pick a promising node first. 

\paragraph{Choosing a constraint.}

The question remains which position constraint~\eqref{eq:1} to
choose for a branching step. We adopt the idea of \emph{strong
  branching} \cite{strong_branching_1995}. The rough idea is to
estimate the progress, \ie,
increase in lower bound, for the residue positions before actually branching on one of them.
This is done by successively fixing each rotamer of a given residue
position and solving the resulting Lagrangian subproblem. Based on the progress
of the single rotamers, we compute an overall score of the residue
position, see below, and pick the one with the highest score as
the next residue position to
branch on. Since the computation time per node of this procedure would
be enormous, we try to estimate 
the locally best residue position by simplifying this search in three different aspects.

\begin{itemize}
\item First, we restrict the evaluation to the most promising residue
  positions. More precisely, we order the residue positions in
  increasing order of their maximum (primal) fractional value of its
  rotamers. We recover the primal solutions by
  taking the convex combinations of the last $k$ solutions $x^t$
  produced in the course of the subgradient optimization. Then, at a
  node of depth $l$ of the branch-and-bound tree, we consider the
  first $\gamma(l)$ percent of the sorted residue positions. Note that
  the concept of a residue position with minimal maximal fractional
  value of its rotamer variables 
  can be
  considered as a generalization of the \emph{most infeasible
    branching} rule for binary variables, to constraints of the
  form~\eqref{eq:1}. 

  \item Second, to estimate the increase of the objective function when
  fixing a rotamer, only a few subgradient iterations are performed,
  along with an aggressive multiplier adjustment. 

\item Finally, we
  choose our scoring function of the residue positions in such a way,
  that they can be computed quickly while still giving a good estimate
  on the overall progress in the dual (lower) bound. 
Let $\mathcal{Q}$ be the subproblem corresponding to the current node of the tree and
let problem $\mathcal{Q}^{j}_{i}$ be obtained from $\mathcal{Q}$ by fixing rotamer $j$ in 
residue $r_i$.
Then the
  score of a residue position $r_i$ is given by
$$\xi(r_i)=\min_{j\in r_i} \Delta_i^j,$$ where $\Delta_i^j=\underline{z}(\mathcal{Q}^{j}_{i})-\underline{z}(\mathcal{Q})$. 
To determine the score of a given
residue position $i$, we test the rotamers in decreasing order of
their fractional values. The goal is to evaluate rotamers $j$ with
small progress $\Delta_i^j$ first, since the subgradient optimization
for the remaining rotamers can be aborted as soon as their increase in
objective function exceeds the smallest progress seen so far. Also
notice that sorting the residue positions as described above is
beneficial for the computation of the residue scores. Whenever we
encounter a rotamer with a progress $\Delta_i^j$ that is smaller than
the smallest progress determined for a previous residue position
$i'\neq i$, we do not have to consider the remaining rotamers of $i'$,
since we are interested in the residue position with maximal score.
\end{itemize}

\paragraph{Choosing a node.}

A primal feasible solution that gives a good upper bound on the minimum total energy is necessary to prune the enumeration tree significantly.
Therefore we follow a depth-first search (DFS) strategy to descend in
the branch-and-bound tree as quickly as possible, increasing the chances of finding a 
new and hopefully better feasible solution. Furthermore, the
Lagrangian subproblems corresponding to a node and to one of its immediate descendants
differ only in one residue position. Therefore, a subproblem can be resolved faster when starting from the multiplier vector $\lambda$ determined in the
immediate parent node. On the downside, once being in a wrong branch, one may spend a long time in this subtree before getting back on a path
leading to an improved solution. We thus combine the advantages of DFS and a best-node first strategy. 
We fix the rotamers of a given residue position in increasing order of
their dual (lower) bounds. Following this approach led to a
considerably smaller number of nodes evaluated in the tree, and we were able to find the optimal 
solution much faster in most of the cases.

\section{Experimental Results}
\label{sec:results}

\renewcommand{\floatpagefraction}{0.75}

We have implemented our combinatorial algorithm for finding an optimal
solution to the Lagrangian relaxation $(LR_\lambda)$ in
C\raisebox{.5ex}{\small ++} using the LEDA and BALL libraries \cite{leda,DBLP:journals/bmcbi/HildebrandtDRBSTMSNMLK10}.
We iteratively improve
the obtained lower bounds on the optimal solution of the \textsc{SCP} problem by applying a standard subgradient approach. In each iteration, we derive from the 
Lagrangian solution a feasible solution to the original problem and
thus an upper bound on the optimal score by evaluating the subgraph
induced by the nodes lying on the shortest path from $V_1$ to $V_k$, see Section~\ref{sec:lagrrel}.
We exploit the upper and lower bounds in a branch-and-bound manner to prune large parts of the search space and to derive a provably optimal solution to the \textsc{SCP} problem.
 
In order to determine an initial upper bound for the branch-and-bound framework, 
we employ a simple local search procedure: Given an initial configuration in which each residue position is assigned the rotamer with the lowest self energy,
residue positions are selected randomly and optimized, \ie, the respective position is assigned the rotamer yielding the best energy within the current conformation.
This minimization proceeds until the energy could not be improved 
several times in a row or a maximum of 100 iterations is reached. 

The only state-of-the-art exact method for SCP that
can cope with protein design instances is the ILP based method
proposed by Kingsford \ea \cite{Kingsford2005}. Available software
packages for DEE or treewidth based
approaches such as R3 \cite{Xie2006} or TreePack \cite{Xu2006} do not
allow several candidate amino acids at each position and are thus not
applicable to protein design instances. Furthermore, our experiments
show that even small protein design instances already have treewidths
of 10 to 20 as compared to 3 to 4 for most homology modeling instances \cite{Xu2006}. Since the complexity
of the TreePack algorithm grows exponentially in the treewidth, a reasonable performance on protein design 
instances is not to be expected. For DEE-based methods, a similar
argument holds because reduced protein design instances are still too
large to be processed in reasonable time by residue unification or other enumeration techniques.
 We therefore compare our Lagrangian based approach only to an implementation that
solves the ILP proposed by Kingsford \ea \cite{Kingsford2005} 
using CPLEX 12.2\footnote{\url{http://www.cplex.com}} with Concert
Technology. 

 
In our experiments, we used two different benchmark sets. The first
set consists of protein design energy files provided by Kingsford \ea~\cite{Kingsford2005}.
It comprises 25 proteins 
with 11 to 124 flexible residue positions. Surface residues are fixed. At each core position up to six different amino acids are allowed. 
The employed energy function comprises statistical potentials and van der Waals interactions.
We omit the experimental results on the simpler homology modeling instances, since almost all of these instances can be solved in a fraction of a second by both our 
Lagrangian relaxation approach and the CPLEX based method. The second
set of protein design instances was taken from Yanover \ea
\cite{Yanover2006}. This set comprises 97 proteins with 40 to 180 amino acids. All residue positions are flexible and at each
position all 20 amino acids are allowed yielding very large problem instances. Here, the more realistic 
Rosetta energy function \cite{Kuhlman2000} was used to determine self and interaction energies.

In a preprocessing phase, we apply
established rules \cite{Goldstein1994,Xie2006} to decrease the size of the problem instances while preserving optimality properties. 
Tables~\ref{cplex_lagrange_times_kingsford} and~\ref{cplex_lagrange_times_yanover} show the running times of our
Lagrangian relaxation branch-and-bound approach and the CPLEX based method using
default settings on the resulting instances on 
a compute cluster with two 2.26 GHz Intel Quad Core processors with 24
GB of RAM on each node, running 64 bit Linux. We applied a time limit
of 12 hours and a memory limit of 16 GB\@. Computations exceeding one of
these limits were aborted.

\begin{table}\small
\begin{center}
\begin{tabular*}{\textwidth}{@{\extracolsep{\fill}}  l  rr |rrr |r |r}
\multicolumn{3}{c}{Instance} & 
\multicolumn{3}{c}{Lagrangian B\&B} &
\multicolumn{1}{c}{CPLEX} & \\
Name & \#res & \#rot & N & H & time/s & time/s & S\\
\hline
1aac & 105 & 1523    & 2 & 1 & 1.73& 3.40 & 2.0\\
1aho & 64 & 981   & 1 & 0 & 0.01 & 0.02 & 2.0\\
1b9o & 123 & 2056    & 3 & 1 & 2.09 & 2.56 & 1.2\\
1c5e & 95 & 1108   & 1 & 0 & 0.12& 0.25 & 2.1\\
1c9o & 66 & 1130  & 2 & 1 & 0.33 & 1.96 & 5.9\\
1cc7 & 72 & 1396    & 1 & 0 & 0.28&0.59 & 2.1\\
1cex & 197 & 2556   & 9 & 2 & 13.37& 33.25 &2.5 \\
1cku & 85 & 1093    & 1 & 0 & 0.03& 0.09 & 3.0\\
1ctj & 89 & 1021    & 1 & 0 & 0.07& 0.27 & 3.9\\
1cz9 & 139 & 2332   & 1 & 0 & 3.7& 18.10 & 4.9\\
1czp & 98 & 1170  & 1 & 0& 0.54 & 4.32 & 8.0\\
1d4t & 104 & 1636    & 1 & 0& 0.37 & 2.36 & 6.4\\
1igd & 61 & 926    & 1 & 0 & 0.01& 0.02 & 2.0\\
1mfm & 153 & 2134  & 25 & 5 & 21.89& 145.63 & 6.7\\
1plc & 99 & 1156    & 2 & 1 & 1.50& 6.08 & 4.1\\
1qj4 & 256 & 4080    & 313 & 10 & 8,424.56&
31,636.40 & 3.8\\
1qq4 & 198 & 2045   & 16 & 4 & 32.56& 38.89 &1.2\\
1qtn & 152 & 2516    & 1 & 0 & 1.50& 3.22 & 2.1\\
1qu9 & 126 & 1817    & 2 & 1 & 0.31& 0.66 & 2.1\\
1rcf & 169 & 2396    & 2 & 1 & 4.76& 12.85 & 2.7\\
1vfy & 67 & 939    & 1 & 0 & 0.01& 0.01 & 1.0\\
2pth & 193 & 3077  & 66 & 6 & 322.28& 518.51 & 1.6\\
3lzt & 129 & 2074   & 7 & 2 & 3.20& 10.64 & 3.3\\
5p21 & 166 & 2874   & 52 & 4 & 106.09& 115.01 & 1.1\\
7rsa & 124 & 1958    & 1 & 0 & 0.78& 3.31 & 4.2\\
\end{tabular*}
\end{center}
\caption{Running times of our Lagrangian relaxation branch-and-bound
  approach and the CPLEX based method on the design instances from \cite{Kingsford2005}. We further give the number of residues (\#res) and the total
number of rotamers (\#rot) of the instance, the number
of nodes (N) and height (H) of the branch-and-bound tree as well as the speedup S
(ratio of running times).}
\label{cplex_lagrange_times_kingsford}
\end{table}
The first three columns
of the table give the characteristics of the instances, \ie, their PDB identifier, the number of residues and the total number
of rotamers. 
The following two columns give characteristics of the branch-and-bound
procedure: Columns N and H give the total number of evaluated nodes
and the height of the branch-and-bound tree, respectively. The remaining
columns give the running times in seconds of our Lagrangian based
approach and the CPLEX based method as well as the ratio of running
times S. Note that we include the time spent in the local search
heuristic for the Lagrangian branch-and-bound approach.

\begin{table}\small
\begin{center}
\begin{tabular*}{\textwidth}{@{\extracolsep{\fill}}  l  rr|rrr |r |r}
\multicolumn{3}{c}{Instance} & 
\multicolumn{3}{c}{Lagrangian B\&B} &
\multicolumn{1}{c}{CPLEX} & \\
Name & \#res & \#rot & N & H & time/s & time/s & S\\
\hline
1brf & 44 & 3524 & 9 & 4 & 293.97 & 469.87 & 1.6\\
1bx7 & 25 & 1048 &  1 & 0 & 0.54 & 5.77 & 10.7\\
1d3b & 66   &  5732  & 1  & 0 &  530.37      & 9,577.68 & 18.1 \\
1en2 & 59 & 2689  & 1& 0 & 19.41 & 39.94 & 2.1\\
1ezg & 58 & 1653    & 2 & 1 & 185.11& 441.23 & 2.4 \\
1g6x & 51 & 3190  & 1 & 0 & 23.96 & 160.64 &6.7\\
1gcq & 65 & 5442 & 4  & 2  &    903.82    & 5,270.08 &9.8\\ 
1i07 & 52 & 3186  & 4 & 1 & 187.45 & 166.20 & 0.9 \\
1kth & 49 & 3330  & 18 & 4 & 798.57& 642.42 & 0.8\\
1rb9 & 43 & 3307    & 7	 & 2 & 127.93& 9,535.72 &74.5
\\
1sem & 54   & 4348   &  192 & 8  &     5,020.55   & 6,470.37 &1.3\\
1vfy & 58 & 3951    & 16 & 2 & 2,540.86& $\dag$ & n/a\\
4rxn & 45 & 3636    & 1 & 0 & 220.33& 3,034.57 &13.8 \\
1a8o & 62& 4510  &6 & 2& 1,418.71 &  $\dag$ & n/a\\
1b67 & 66 & 5543  &27 & 4& 3,822.09 &  $\dag$ & n/a\\
1bbz & 52 & 3935 & 7 & 2 &1,329.26 &  $\dag$ & n/a\\
1bf4 & 60 & 5289   & 12 & 3 &1,875.32 &  $\dag$ & n/a\\
1c75 & 63 & 4323  & 25 & 2 & 7,175.69&  $\dag$ & n/a\\
1cc8 & 69 & 6515   & 26 & 2 &16,508.10 &  $\dag$ & n/a\\
1d3b & 66 & 5732  & 1 & 0 & 530.37& $\dag$ & n/a\\
1fr3 & 61 & 5100 & 22 & 4 & 6,997.76&  $\dag$ & n/a\\
1gut & 62 & 4945  & 22 & 2 & 7,745.17&   $\dag$ & n/a\\
1hg7 & 65 & 5047  & 5 & 2 & 987.51&  $\dag$ & n/a\\
1i27 & 69 & 5934   & 39 & 6 & 4,070.20&  $\dag$ & n/a\\
1igd & 60 & 5207  & 18 & 4 & 3,163.14&  $\dag$ & n/a\\
1igq & 53 & 4582 & 18 & 5 & 4,294.16&  $\dag$ & n/a\\
1iqz & 75 & 5412   & 15 & 2 & 2,137.58&  $\dag$ & n/a\\
1j75 & 55 & 4861  & 14 & 4 & 5,704.83&  $\dag$ & n/a\\
1jo8 & 54 & 4680  & 41 & 4 & 1,830.23&  $\dag$ & n/a\\
1kq1 & 58 & 5244   & 7 & 1 & 3,990.65&  $\dag$ & n/a\\
1l9l & 69 & 5518   & 4 & 2 & 1,514.50& $\dag$ & n/a\\
1ldd & 71 & 6383   & 8 & 1 & 4,582.84&  $\dag$ & n/a\\
1ljo & 69 & 6428    & 14 & 3 & 5,624.10&  $\dag$ & n/a\\
1mhn & 53 & 4454   & 3 & 2 & 570.15 &  $\dag$ & n/a\\
1nkd & 56 & 4148     & 9 & 4 & 1,119.56&  $\dag$ & n/a\\
1oai & 56 & 4330    & 73 & 3 & 20,021.10&  $\dag$ & n/a\\
1plc & 92 & 7955    & 34 & 4 & 24,308.50&   $\dag$ & n/a\\
1pwt & 58 & 4876     & 2 & 1 & 886.94&   $\dag$ & n/a\\
1r69 & 60 & 4926    & 49 & 2 & 27,862.50&  $\dag$ & n/a\\
1wap & 65 & 5551    & 14 & 3 & 5,267.68&  $\dag$ & n/a\\
2igd & 59 & 5262  & 13 & 3 & 6,332.26&   $\dag$ & n/a\\
1c4q & 65 & 5598   & 843 & 10 & 30,789&   $\dag$ & n/a\\
1c9o & 60 & 5305   & 97 & 8 & 11,627.20&   $\dag$ & n/a\\
1ctj& 84 & 6232   & 265 & 10 & 7,083.65&   $\dag$ & n/a\\
1dj7& 69 & 5571   & 22 & 4 & 12,581.90&   $\dag$ & n/a\\
1e0b& 58 & 4715   & 147 & 6 & 30,840.10&   $\dag$ & n/a\\
1erv& 101& 9150    & 76 & 8 & 31,539.70&   $\dag$ & n/a\\
1fk5& 81 & 5714     & 28 & 5 & 14,625.30&   $\dag$ & n/a\\
1g2b& 59 & 4926    & 129 & 8 & 6,880.45&   $\dag$ & n/a\\
1mgq& 71& 6250    & 4 & 1 & 2,613.06&   $\dag$ & n/a\\
1vie& 56& 4803    & 3 & 1 & 866.54&   $\dag$ & n/a\\
\end{tabular*}
\end{center}
\caption{Running times of our Lagrangian relaxation
  branch-and-bound approach and the CPLEX based method on the design instances
  from \cite{Yanover2006}. Instances which cannot be solved by both
  approaches within a time limit of 12 hours are omitted. The $\dag$
  sign indicates that a computation exceeded either the time limit (12~h) or the memory limit (16~GB).}
\label{cplex_lagrange_times_yanover}
\end{table}


On the first dataset, our Lagrangian based approach outperforms the
state-of-the-art CPLEX based method on all 25 instances. The small number of nodes evaluated in the course
of the branch-and-bound procedure indicates sharp lower and upper
bounds derived from the Lagrangian solutions. On the more challenging
second dataset, our method could solve 52 of the 97 instances within
12~h, whereas the CPLEX based method could only finish 12 instances
within the time and memory limits. 

\section{Conclusions and Outlook}
\label{conclusion}

We have constructed a Lagrangian relaxation of the Kingsford ILP formulation of the \textsc{SCP} problem that allowed us to obtain strong bounds by solving 
a modified shortest path problem on the underlying $k$-partite
graph. By utilizing these bounds within a branch-and-bound framework we achieved running times that
outperform a state-of-the-art exact method that uses the professional
mathematical programming solver CPLEX\@. Our implementation of the Lagrangian branch-and-bound approach as well
as the data sets used in this paper are freely available as the
package \textsc{scp} of the planet lisa software library~\cite{planetlisa}.

Future work on exact side-chain placement should explore
possible connections to the recently introduced method by Sontag
\ea \cite{Sontag2008}, which is based on belief propagation. This heuristic algorithm is currently the best
non-exact method and finds optimal solutions astonishingly often. In
our opinion, underlying ideas from the area of belief
propagation may be useful also in a truly exact method.

The mathematical model of the \textsc{SCP} problem as studied in this work appears in a wide range of applications including image understanding
, error correcting codes
, and frequency assignment in telecommunications. We believe that our approach can be applied successfully in these areas, too.


\paragraph{Acknowledgements.} 
Much of the work has been carried out while NCT visited CWI on a CWI
internship. The authors thank Inken Wohlers and Oliver Kohlbacher for
useful discussions. Computational experiments were sponsored by the NCF for the use of supercomputer facilities, with financial support from NWO.

\bibliographystyle{abbrv}
\bibliography{sally}


\end{document}